\documentclass[11pt]{article}

\usepackage[small]{titlesec}
\usepackage[small,bf]{caption}

\usepackage{amsfonts}
\usepackage{amsmath}
\usepackage{amssymb}

\usepackage{algorithm}
\usepackage[noend]{algpseudocode}
\usepackage{graphicx}
\usepackage{tabularx}
\usepackage{times}
\usepackage{url}

\newtheorem{theorem}{Theorem}
\newtheorem{lemma}{Lemma}

\newcommand{\sq}{\hbox{\rlap{$\sqcap$}$\sqcup$}}
\newcommand{\qed}{\hspace*{\fill}\sq}
\newenvironment{proof}{\noindent {\bf Proof.}\ }{\qed\par\vskip 4mm\par}

\begin{document}

\begin{titlepage}

\title{
	\textbf{Infinite Object Coating in the Amoebot Model}
}

\author{
	Zahra Derakhshandeh$^1$,
	Robert Gmyr$^2$,
	Andr\'ea W.\ Richa$^1$,\\
	Christian Scheideler$^2$,
	Thim Strothmann$^2$,
	Shimrit Tzur-David$^3$\\
	\\
	$^1$ Department of Computer Science and Engineering,\\
	Arizona State University, USA\\
	\\
	$^2$ Department of Computer Science,\\
	University of Paderborn, Germany\\
	\\
	$^3$ Department of Computer Science,\\
	Ben-Gurion University, Israel
}

\date{}

\maketitle \thispagestyle{empty}

\begin{abstract}
The term \emph{programmable matter} refers to matter which has the ability to change its
physical properties (shape, density, moduli, conductivity, optical properties,
etc.) in a programmable fashion, based upon user input or autonomous sensing.
This has many applications like smart materials,
autonomous monitoring and repair, and minimal invasive
surgery. While programmable matter might have
been considered pure science fiction more than two decades ago, in recent years a large amount of research
has been conducted in this field.
Often programmable matter is envisioned as a very large number of small 
locally interacting computational \emph{particles}.
We propose the Amoebot model, a new model which builds upon this vision of 
programmable matter. Inspired by the behavior of amoeba, the 
Amoebot model offers a versatile framework to model self-organizing
particles and facilitates rigorous algorithmic research in the area of programmable matter.
 We present an algorithm for the problem of coating an
infinite object under this model, and prove the correctness of the algorithm
 and that it is work-optimal.
 
\end{abstract}

\end{titlepage}

\section{Introduction}
Recent advances in microfabrication and cellular engineering foreshadow that in the next few decades it might be
possible to assemble myriads of simple information processing units at
almost no cost. Microelectronic mechanical components have become so inexpensive to
manufacture that one can anticipate integrating logic circuits, microsensors,
and communications devices onto nano-computational components.
 Imagine coating bridges and buildings with smart paint that senses and reports
on traffic and wind loads and monitors structural integrity. A smart-paint
coating on a wall could sense vibrations, monitor the premises for intruders,
and cancel noise. There has also been amazing progress in understanding the
biochemical mechanisms in individual cells such as the mechanisms behind cell
signaling and cell movement \cite{AE07}. 
Recently, it has been demonstrated that, in
principle, biological cells can be turned into finite automata \cite{BPEA+01} or even
pushdown automata \cite{KSP12}, so one can imagine that some day one can tailor-make
biological cells to function as sensors and actuators, as programmable
delivery devices, and as chemical factories for the assembly of nano-scale
structures. 

One can envision producing vast quantities of individual microscopic computational particles---whether microfabricated
particles or engineered cells---to form {\em programmable matter}, as coined by Toffoli and Margolous \cite{to91}.
These particles are possibly faulty, sensitive to the environment, and may produce various types
of local actions that range from changing their internal state to communicating with
other particles, sensing the environment, moving to a different location,
changing shape or color, or even replicating. Those individual
local actions may then be used to change
the physical properties, color, and shape of the matter at a global scale. 
 
We propose \emph{Amoebot}, a new amoeba-inspired model for programmable matter\footnote{A preliminary version of our model was presented at the First Biological Distributed Algorithms (BDA) Workshop, co-located with DISC, October 2013, and  has appeard as a Brief Announcement at ACM SPAA 2014~\cite{DolevGRS13, DBLP:conf/spaa/DerakhshandehDGRSS14}.}.
In our model, the programmable matter consists of particles that can bond to neighboring particles
and use these bonds to form connected structures.
Particles only have local information and have modest computational power:
Each particle has only a constant-size memory and behaves similarly to a finite state machine.
The particles act asynchronously and they achieve locomotion by expanding and contracting,
which resembles the behavior of amoeba~\cite{AE07}. 

%
%
%

\subsection{Our Contributions}
\vspace{-2mm}
Our proposed Amoebot model, presented in Section~\ref{sec:model}, offers a versatile framework to model self-organizing
particles and facilitates rigorous algorithmic research in the area of programmable matter. In addition, we present an algorithm for the problem of coating an
infinite object under this model in Section~\ref{sec:algorithm}, and prove the correctness of the algorithm (Theorem~\ref{thm:solve}) and that the algorithm is work-optimal (Theorem~\ref{thm:work}). 

\section{Related Work}
\vspace{-2mm}
While programmable matter may have seemed like 
science fiction more than two decades ago, we have seen many advances in
this field recently.
One can distinguish
between active and passive systems. In passive systems the particles either do not have any
intelligence at all (but just move and bond based on their structural
properties or due to chemical interactions with the environment), or they
have limited computational capabilities but cannot control their
movements. Examples of research on {\em passive systems} are DNA computing \cite{Adl94, BDLS96, CDBG11, DPSS11, NKC03, WLWS98}, tile self-assembly systems in general \cite{st11, st08}, 
population protocols \cite{AAD+06}, and slime molds \cite{BMV12, LTT+10, WTTN11}.
We will not describe these models in detail as they are only of little relevance for our approach.
On the other hand in \emph{active systems}, there are computational
particles that can control the way they act and move in order to solve a specific task.
Self-organizing networks, robotic swarms, and modular robotic systems are 
some examples of active systems.

\textit{Self-organizing networks} have been studied in many different contexts.
Networks that evolve out of local, self-organizing behavior have been heavily studied in the context of
\textit{complex networks} such as small-world networks \cite{BA99, Kle00, WS98}.
However, whereas a common approach for the complex networks field is to study
the global effect of given local interaction rules,
we aim at developing local interaction rules in order to obtain a desired global effect.

In the area of \textit{swarm robotics} it is usually assumed that there is a collection of autonomous robots
that have limited sensing, often including vision, and communication ranges, and that can freely move in a given area. 
They follow a variety of goals, as for example graph exploration (e.g., \cite{fl13}), gathering problems (e.g., \cite{AG3, ci12})
, and shape formation problem (e.g., \cite{fl08}).
Surveys of recent results in swarm robotics can be found in \cite{Ker12, McL08}; other samples of representative work can be found in e.g., ~\cite{FGK10, BFMS11, DFSY10
, DS08, CP08
, Kat05, HABFM02, SY99, FS10, PZ06, KM11, AR10, RS10}.
 Besides work on how to set up robot swarms in order
to solve certain tasks, a significant amount of work has also been invested in
order to understand the global effects of local behavior in natural swarms
like social insects, birds, or fish (see e.g.,~\cite{Cha09, DBLP:1211-1909}).
While the analytical techniques developed in the area of swarm robotics and natural swarms are of some relevance for this work,
the underlying model differs significantly as we do not allow free movement of particles.

While swarm robotics focuses on inter-robotic aspects in order to perform certain tasks,
the field of \textit{modular self-reconfigurable robotic systems} focuses on intra-robotic aspects
such as the design, fabrication, motion planning, and control of autonomous kinematic machines with variable morphology 
(see e.g., \cite{FNKB88, YSS+07}).
\textit{Metamorphic robots}  form a subclass of self-reconfigurable robots that shares the characteristics of our model
that all particles are identical and that they fill space without gaps \cite{Chi94}.
The hardware development in the field of self-reconfigurable robotics has been complemented
by a number of algorithmic advances (e.g., \cite{BKRT04, WWA04}),
but so far 
mechanisms that scale to hundreds or thousands of individual units are still under investigation,
and no rigorous theoretical foundation is available yet.

As in our model, the work in \cite{em13} also assumes that each system particle is a finite state machine with constant-size memory
operating in an asynchronous distributed fashion. However, the work in \cite{em13} assumes a {\em static} network
topology and the problems considered only address what can be {\em computed} in an asynchronous network of randomized constant-size memory finite state machines, and
not how desired topology/shape can be achieved in systems where finite state machine particles can move (in addition to computing).
 
The \emph{nubot} model~\cite{winfree13}, by Woods et al., was developped independently to our model (our model was originally presented at the First Workshop on Biological Distributed Algorithms (BDA), October 2013~\cite{DolevGRS13, DBLP:conf/spaa/DerakhshandehDGRSS14}), and aims at providing the theoretical framework that would allow for
a more rigorous algorithmic study of biomolecular-inspired systems,
more specifically of self-assembly systems with active molecular
components. While our model shares many similarities with the nubot model at a high level, many of the assumptions underlying the nubot model are different
from ours: For example, in the nubot model, particles are allowed to
replicate at will and are allowed to drag many other particles as they
move in space (which is pertinent in many molecular level systems, but
which would not be feasible in systems with very large numbers of
nano-robots with weak bond structures); also the number of states (and
hence also the memory size) that a particle can be in is proportional
to $\log n$---where $n$ relates to the number of particles in the
final desired configuration the system should assume---whereas in
our model the number of possible states is assumed to be constant.

\newcommand{\junk}[1]{}
\junk{
While these systems (and hence the nubot model) share many common characteristics with the programmable matter systems we consider in this work (they also assume an asynchronous network model of simple state machines, they also target shape formation problems, etc.), there are significant differences that make the nubot model not applicable to the scenarios we consider. For example, in the nubot model,  particles are allowed to replicate at will and at high rates, and particles are allowed to drag other particles as they move in space (which is pertinent in many molecular level systems, but which would pose a problem in systems with very large numbers of nano-robots with weak bond structures).
More importantly, the nubot model assumes a system of computationally-limited particles, but {\em not} a system of finite-state machines: They assume that each particle has a memory of size proportional to $\log n$, where $n$ is the number of particles in the final desired configuration the system should assume. This assumption is necessary in the nubot model, since they consider shape formation problems with predetermined dimensions, but it also makes the systems considered much more computationally powerful than the ones we consider in this paper (e.g., particles can individually count to $n$, which would be impossible in a system of finite state machines) .
} 

\section{Model}
\label{sec:model}
\vspace{-2mm}
Consider the \emph{equilateral triangular graph} $G_{eqt}$, see Figure\ref{fig:graph}.
\begin{figure}[b!]
    \centering
    \includegraphics[scale=0.85]{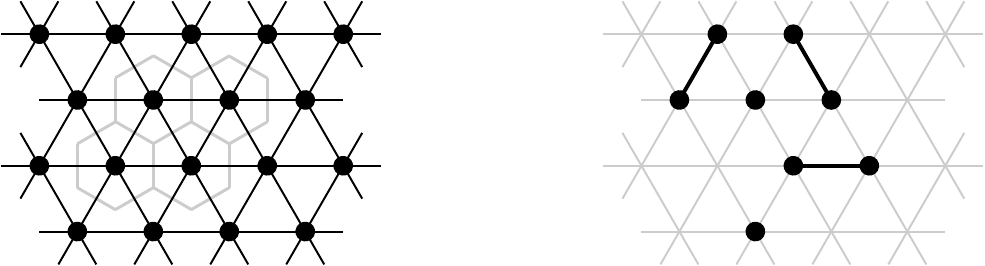}
    \caption{\small
    	The left half of the figure depicts a section of the infinite equilateral triangular graph $G_{eqt}$.
    	Nodes are shown as black circles.
    	For four of the nodes the dual faces in the hexagonal tiling of the Euclidean plane are shown in gray.
    	The right half shows five particles on the graph.
    	When depicting particles we draw the graph as a gray mesh without nodes.
    	A particle occupying a single node is depicted as a black circle,
    	and a particle occupying two nodes is depicted as two black circles connected by an edge.
    }
    \label{fig:graph}
\end{figure}
A \emph{particle} occupies either a single node or a pair of adjacent nodes in $G_{eqt}$,
and every node can be occupied by at most one particle.
Two particles occupying adjacent nodes are defined to be \emph{connected}
and we refer to such particles as \emph{neighbors}.
The graph $G_{eqt}$ is the dual graph of the hexagonal tiling of the Euclidean plane
as indicated in Figure~\ref{fig:graph}.
So geometrically the space occupied by a particle is bound
by either one face or two adjacent faces in this tiling of the plane.

Every particle has a \emph{state} from a finite set $Q$.
Connected particles can communicate via the edges connecting them in the following way.
A particle $p$ holds a \emph{flag} from a finite alphabet $\Sigma$ for each edge that is incident to $p$
(i.e., all edges incident to a node occupied by $p$
except the edge between the occupied nodes if $p$ occupies two nodes).
A particle occupying the node on the other side of such an edge can read this flag.
This communication process can be used in both directions over an edge.
In order to allow a particle $p$ to address the edges incident to it,
the edges are labeled from the local perspective of $p$.
This labeling starts with $0$ at an edge leading to a node
that is only adjacent to one of the nodes occupied by $p$
and increases counter-clockwise around the particle.

Particles move through \emph{expansion} and \emph{contraction}:
If a particle occupies one node, it can expand into an unoccupied adjacent node to occupy two nodes.
If a particle occupies two nodes, it can contract out of one of these nodes to occupy only a single node. (Those two actions can be naturally physically realized on the dual hexagonal tiling of the Euclidean plane.)
Accordingly, we call a particle occupying a single node \emph{contracted}
and a particle occupying two nodes \emph{expanded}.
Note that we can identify six directions in our graph corresponding to the directions of the six edges incident to a node.
The direction of the edge labeled $0$ is defined to remain constant throughout all movement.
We call this direction the \emph{orientation} of a particle.
Figure~\ref{fig:movement} shows an example of the movement of a particle.
\begin{figure}[b!]
    \centering
    \includegraphics[scale=0.85]{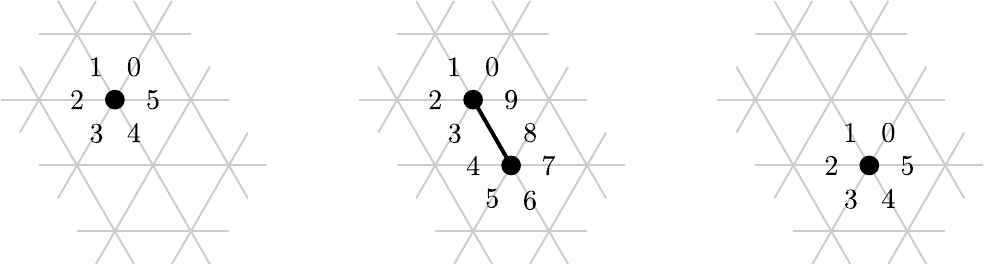}
    \caption{\small
        The three parts of the figure show a moving particle
        together with the labels seen by the particle.
        On the left, the particle occupies only a single node.
        The particle then expands in the direction of the edge labeled $4$
        resulting in the particle occupying two nodes as depicted in the middle.
        Since the expansion changes the number of edges incident to the particle, the edges have to be relabeled.
        The direction of the edge labeled $0$ remains constant.
        Next the particle contracts out of one of the nodes it currently occupies
        towards the direction of the edge labeled $6$
        resulting in the particle occupying only a single node as depicted on the right.
        Again, the edges incident to the particle are relabeled.
    }
    \label{fig:movement}
\end{figure}
Besides executing expansions and contractions in isolation,
we allow pairs of connected particles to combine these primitives to perform a coordinated movement:
One particle can contract out of a certain node at the same time as another particle expands into that node.
We call this movement a \emph{handover}, see Figure~\ref{fig:handover}.
\begin{figure}
    \centering
    \includegraphics[scale=0.85]{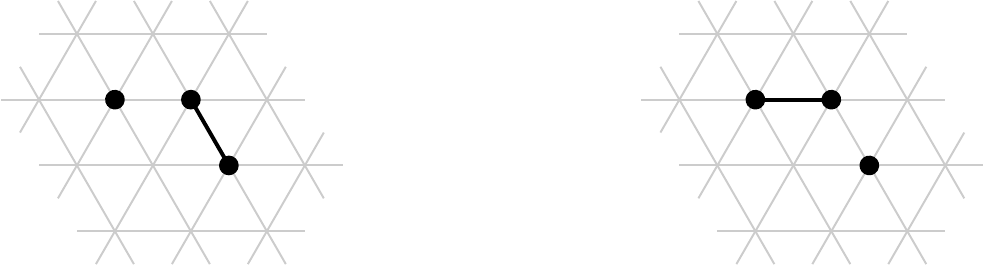}
    \caption{\small
       Two particles performing a handover.
    }
    \label{fig:handover}
\end{figure}
The particles involved in a handover are defined to remain connected during its execution.

Computationally, particles resemble finite state machines.
A particle acts according to a \emph{transition function}
\[
    \delta: \; Q \times \Sigma^{10} \; \to \; \mathcal{P}(Q \times \Sigma^{10} \times M).
\]
For a particle $p$ the function takes the current state of $p$
and the flags $p$ can read via its incident edges as arguments.
Here, the $i$-th coordinate of the tuple $\Sigma^{10}$ represents the flag read
via the edge labeled $i$ when numbering the coordinates of the tuple starting at $0$.
If for a label $i$ there is no edge with that label or if the respective edge leads to a node that is not occupied,
the coordinate of the tuple is defined to be $\varepsilon$.
The value $\varepsilon \in \Sigma$ is reserved for this purpose and cannot be set as a flag by a particle.
The transition function maps to a set of \emph{turns}.
A turn is a tuple specifying a state to assume, flags to set, and a movement to execute.
The set of movements is defined as
\[
    M =
    \{ \text{idle} \} \; \cup \;
    \{ \text{expand}_i \mid i \in [0, 9] \} \; \cup \;
    \{ \text{contract}_i \mid i \in [0, 9] \} \; \cup \;
    \{ \text{handoverContract}_i \mid i \in [0, 9] \}.
\]
The movement \emph{idle} means that $p$ does not move,
and \emph{expand$_i$} and \emph{contract$_i$} are defined as mentioned above.
The index $i$ specifies the edge that defines the direction along which the movement should take place,
as shown in the example in Figure~\ref{fig:movement}. Note that there are only two possible directions for
a contract operation. 
The movement \emph{handoverContract$_i$} specifies a contraction that can only be executed as part of a handover.
In summary, a transition function specifies a set of turns a particle would like to perform
based on the locally available information.

A system of particles progresses by executing atomic \emph{actions}.
An action is either the execution of an isolated turn for a single particle
or the execution of a turn for each of two particles resulting in a handover between those particles.
Note that if a movement is not executable, the respective action is not enabled:
For example, a particle occupying two nodes cannot expand although it might specify this movement in a turn.
As another example, a particle cannot expand into an occupied node except as part of a handover.
Finally, an action consisting of an isolated turn involving the movement \emph{handoverContract$_i$}
is never enabled as this movement can only be performed as part of a handover.
The transition function is applied for each particle to determine the set of enabled actions in the system.
From this set, a single action is arbitrarily chosen and executed.
The process of evaluating the transition function and executing an action continues
as long as there is an enabled action.

Two actions are said to be independent if they do not involve nodes that are neighbors in $G_{eqt}$. Each particle can locally ensure that at most one action in its neighborhood is executed at any point in time. Hence, all of our results also hold if a set of mutually independent actions was chosen to be concurrently executed at any point in time.

\section{Morphing Problems}
\label{sec:morphingProblems}
\vspace{-2mm}
In general, we define a \emph{morphing problem} as being a problem 
in which a system of particles has to morph into a shape with specific characteristics 
(by changing the positioning of the particles in $G_{eqt}$) while sustaining connectivity.
Examples of morphing problems are the formation of geometric shapes and
coating objects (i.e., surrounding a given set of nodes).
Before we can formally define morphing problems, we need some definitions.

We define the \emph{configuration of a particle} as
the tuple of its state, its flags, the set of nodes it occupies, and its orientation.
A system of particles progresses by performing atomic actions,
each of which affects the configuration of one or two particles.
Therefore, a system progresses through a sequence of configurations
where a \emph{configuration of a system} is the set of configurations of all its particles.
We define the \emph{connectivity graph} $G(c)$ of a configuration $c$ as
the subgraph of $G_{eqt}$ induced by the occupied nodes in $c$.

We can formally define a morphing problem as a tuple $M = (I, G)$
where $I$ and $G$ are sets of connected configurations.
We say $I$ is the set of initial configurations and $G$ is the set of goal configurations.
An algorithm $A$, formally defined by a transition function $\delta$,
\emph{solves} $M$ if three conditions hold:
Consider the execution of $A$ on a system in an arbitrary configuration from $I$.
First, the system stays connected throughout the execution of $A$.
Second, the execution eventually reaches a configuration in which the transition function of each particle maps to the empty set
 (we say $A$ \emph{terminates}).
Third, when the execution terminates, the reached configuration is from $G$.

\section{Infinite Object Coating}
\vspace{-2mm}
As a subclass of the class of morphing problems, one can consider \emph{coating problems}
in which an object is to be coated (i.e., surrounded or engulfed)
by the particles of a system as uniformly as possible.
We investigate the \emph{Infinite Object Coating} problem where the object has an infinite surface and,
accordingly, a uniform coating is accomplished when all the particles of a system are directly connected to the object.

\subsection{Problem Definition}
\label{sec:problemDefinition}
\vspace{-2mm}
In the Infinite Object Coating problem, an object can be represented by a set of contracted particles
 occupying nodes in $G_{eqt}$. 
These particles are in a special \emph{object state}, and we refer to these particles as \emph{object particles}.
A transition function must map to the empty set for an object particle,
and no particle can switch into the object state.
We denote the number of non-object particles in a system by $n$.
In the reminder of this paper, unless otherwise stated, when we refer to a particle, we mean a non-object particle.
We say a particle \emph{lies on the surface} of the object if it is connected to an object particle.
Consider a connected induced subgraph $C$ of $G_{eqt}$.
The subgraph $C$ is called \emph{compact} if $G_{eqt} - C$ is $2$-connected.
An object that induces a compact subgraph in $G_{eqt}$ is a \emph{valid object}.
Intuitively, this definition means that a valid object cannot have \emph{tunnels} of width one,
see Figure~\ref{fig:tunnel}.
Disallowing these tunnels allows particles to move along the object in single file
(as will be described in Section~\ref{sec:movingAlongASurface}) without blocking each other
and therefore improves the clarity of presentation by avoiding boundary cases.

\begin{figure}[hb]
    \centering
    \includegraphics[scale=0.85]{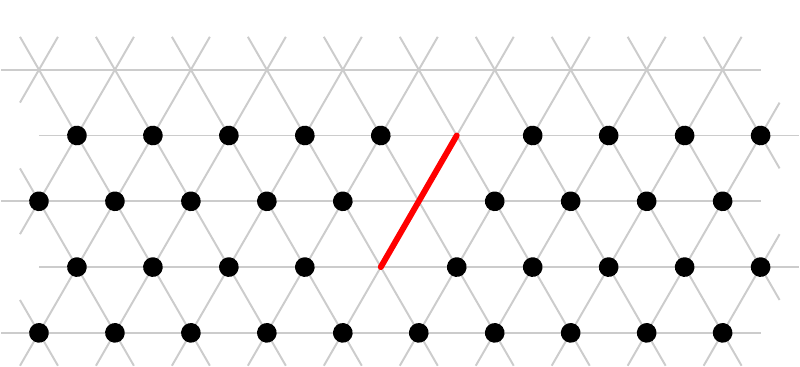}
    \caption{\small
       An example of an invalid object.
       The object occupies a half-plane except for the nodes marked by the solid line.
       These nodes form a tunnel of width one in the object
       since removing the topmost of these nodes from $G_{eqt} - C$ disconnects the two remaining nodes in that graph.
    }
    \label{fig:tunnel}
\end{figure}

As presented in Section~\ref{sec:morphingProblems}, a morphing problem is defined as a tuple $(I, G)$
where $I$ is a set of initial configurations and $G$ is a set of goal configurations.
For the Infinite Object Coating problem, $I$ is the set of all connected configurations consisting of
a valid object together with a finite set of contracted particles.
Every particle stores a \emph{phase} as part of its state and in an initial configuration every particle is in an \emph{inactive} phase; we will elaborate on phases in Section~\ref{sec:spanningForestAlgorithm}.
Similarly, the set $G$ is the set of all configurations consisting of
a valid object together with a finite set of contracted particles that all lie on the surface of the object.

The general Amoebot model as described in Section~\ref{sec:model} can take various specific forms depending on
how systems of particles are initialized, what information particles keep track of in their state,
and what information they share over their edges.
In the Infinite Object Coating problem, we do not make any assumption about the orientation of the particles.
Therefore, we work in a no-compass variant of the model.
While particles do not share a common sense of direction,
they are able to keep track of directions by storing edge labels in their state
and updating these labels upon movement.
The updates can be encoded in the transition function.
We assume that a particle keeps track of whether it is contracted or expanded; 
the particle also keeps track of which edge labels are incident to the occupied node that is an endpoint of the edge labeled $0$.
Additionally, we assume that for an edge with label $i$ the corresponding flag always includes the index $i$,
the information whether the edge is incident to the occupied node that is an endpoint of the edge labeled $0$,
and whether the particle is contracted or expanded.
Using this information, a particle that reads a flag can compare its orientation to the orientation of the particle
that set the flag and therefore particles can exchange information about directions.
Lastly, we assume that a particle keeps track of what edge labels specify valid contraction indices.

For the sake of generality, the Amoebot model does not enforce any fairness condition on the execution of actions.
However, for the Infinite Object Coating problem we make the following assumption:
Any set of of consecutive configurations in which a particle $p$ could be affected by an enabled action, but is not, is finite.

\subsection{Algorithm}
\label{sec:algorithm}
\vspace{-2mm}
Our algorithm for the Infinite Object Coating problem is a combination of three algorithmic primitives.
First, particles lying on the surface of the object lead the way by moving in a common direction along the surface.
Second, the remaining particles follow behind the leading particles
resulting in the system flattening out towards the direction of movement.
Third, particles on the surface check whether there are particles not lying on the surface
and use this information to eventually achieve termination.
We present each of these primitives in detail in the following sections.

\subsubsection{Moving Along a Surface}
\label{sec:movingAlongASurface}
\vspace{-2mm}
Our first algorithmic primitive solves a simple problem:
We want all particles on the surface to move along the surface in a common direction.
However, before we can present our algorithm for this problem, we need some definitions.
For an expanded particle, we denote the node the particle last expanded into as the \emph{head} of the particle
and call the other occupied node its \emph{tail}.
For a contracted particle, we define the single occupied node to be both the head and the tail.
The set of labels associated to the edges incident to the head can be encoded as part of the state,
and this information can be set upon expansion as part of the transition function.
Therefore, a particle can always distinguish the labels of edges incident to its head (\emph{head labels})
from the labels of edges incident to its tail (\emph{tail labels}).
Accordingly, we call edges that are labeled with a head label \emph{head edges}
and the remaining edges \emph{tail edges}.
Combined with the information about valid contraction indices described in Section~\ref{sec:problemDefinition},
a particle can deliberately contract out of its tail or its head.
In our algorithm, particles are only allowed to contract out of their tails
so that the fact that a particle contracts uniquely defines the contraction direction.
Note that with this convention, the head of a particle still is occupied by that particle after a contraction.

A particle on the surface can move along the surface in two directions.
However, we want all particles on the surface to move in a common direction.
The simple procedure given in Algorithm~\ref{alg:FaSAlgorithm},
which is very similar to the moving algorithm presented by Drees et al.~\cite{DHKS12},
can be used to achieve this goal:
\begin{algorithm}
	\begin{algorithmic}
		\State let $k$ be the size of the neighborhood of the particle (i.e., $k = 6$ or $k = 10$)
		\State let $i$ be the label of an edge connected to the object
		\While{edge $i$ is connected to the object} 
		\State $i \; \gets \; (i - 1) \mod k$
		\EndWhile
		\State \textbf{return} $i$
	\end{algorithmic}
	\caption{Movement along a Surface}
	\label{alg:FaSAlgorithm}
\end{algorithm}
A contracted particle uses the procedure to compute the direction of an expansion,
and an expanded particle simply contracts according to above definitions.
The correctness of this approach is based on two facts.
First, all particles share a common sense of rotation
(i.e., the edge labels always increase counter-clockwise around a particle).
Second, according to our definition of a valid object,
an object must occupy a single consecutive sequence of nodes around a particle
from the local perspective of that particle and not all nodes around a particle belong to the object.

\subsubsection{Spanning Forest Algorithm}
\label{sec:spanningForestAlgorithm}
\vspace{-2mm}
As the name suggests, the \emph{spanning forest algorithm} aims to organize the particles in a system
as a spanning forest where the particles that represent the roots of the trees in the forest
are considered leaders whom the remaining particles follow.
Therefore, the movement of the system is dominated by the movement of the leaders.
Every particle that is connected to the surface becomes a leader,
and leaders move along the surface as described in the previous section.
Algorithm~\ref{alg:spanningForestAlgorithm} provides a detailed description of this approach.

\begin{algorithm}
	A particle is in one of three phases \emph{inactive}, \emph{follow}, and \emph{lead}.
	Initially, all particles are assumed to be in phase \emph{inactive}.
	The phase of a particle is encoded as part of its state,
	and a particle indicates its phase as part of all its flags.
	We call a particle in phase follow a \emph{follower} and a particle in phase lead a \emph{leader}.
	A follower stores a head label $d$ in its state
	and includes a \emph{follow indicator} in the flag for the edge with label $d$.
	Depending on its phase, a particle $p$ behaves as described below.
	The transition function maps either to a set containing a single turn or to the empty set.
	The specified conditions are to be checked in the given order.
	If a condition holds, the transition function maps to the set containing only the respective turn.
	If none of the conditions holds, the transition function maps to the empty set.
	\begin{tabularx}{\textwidth}{lX}
		\textbf{inactive}: &
		If $p$ is connected to the surface, it becomes a leader and executes the idle movement.
		If an adjacent node is occupied by a leader or a follower,
		$p$ sets $d$ to point towards that node and becomes a follower.
		\\
		\textbf{follow}: &
		If $p$ is contracted and connected to the surface,
        it becomes a leader and executes the idle movement.
		If $p$ is contracted and there is an expanded particle $p'$ occupying the node
		reached via the edge labeled $d$, $p$ expands in direction $d$ as an attempt to a handover and sets $d$
		to correspond to the contraction direction of $p'$.
		If $p$ is expanded and a follow indicator is read from a contracted neighbor over a tail edge,
		$p$ executes a handover contraction and changes $d$ to keep the direction constant.
		If $p$ is expanded, no follow indicator is read over a tail edge,
		and $p$ has no inactive neighbor, $p$ contracts and changes $d$ to keep the direction constant.
		\\
		\textbf{lead}: &
		If $p$ is contracted, it expands in the direction computed by Algorithm~\ref{alg:FaSAlgorithm}.
		If $p$ is expanded and a follow indicator is read from a contracted neighbor over a tail edge,
        $p$ executes a handover contraction.
		If $p$ is expanded, no follow indicator is read over a tail edge, and there is no inactive neighbor,	$p$ contracts.
	\end{tabularx}
	\caption{Spanning Forest Algorithm}
	\label{alg:spanningForestAlgorithm}
\end{algorithm}

In contrast to particles in phase inactive, we say followers and leaders are \emph{active}.
As specified in Algorithm~\ref{alg:spanningForestAlgorithm}, the value $d$ is only defined for followers.
We denote the node in $G_{eqt}$ reached from a follower $p$ via the edge labeled $d$ as $u(p)$.
The following lemmas demonstrate some properties that hold during the execution of the spanning forest algorithm
and will be used in Section~\ref{sec:analysis} to analyze our complete algorithm.

\begin{lemma}
    \label{lem:successor}
	For a follower $p$ the node $u(p)$ is occupied by an active particle.
\end{lemma}
\begin{proof}
    Consider a follower $p$ in any configuration during the execution of Algorithm~\ref{alg:spanningForestAlgorithm}.
    Note that $p$ can only get into phase follow from phase idle,
    and once it leaves the follow phase it will not switch to that phase again.
    Consider the first configuration $c_1$ in which $p$ is a follower.
    In the configuration $c_0$ immediately before $c_1$, $p$ must be inactive
    and it becomes a follower because of an active particle $p'$ occupying $u(p)$ in $c_0$.
    The particle $p'$ still occupies $u(p)$ in $c_1$.
    Now assume that $u(p)$ is occupied by an active particle $p'$ in a configuration $c_i$,
    and that $p$ is still a follower in the next configuration $c_{i+1}$ that results from executing an action $a$.
    If $a$ affects $p$ and $p'$, the action must be a handover in which $p$ updates its value $d$
    such that $u(p)$ changes but $p'$ again occupies $u(p)$ in $c_{i+1}$.
    If $a$ affects $p$ but not $p'$,
    it must be a contraction in which $u(p)$ does not change and is still occupied by $p'$.
	If $a$ affects $p'$ but not $p$, there are multiple possibilities.
 	The particle $p'$ might switch from phase follow to phase lead or it might expand,
 	neither of which violate the lemma.
 	Furthermore, $p'$ might contract.
 	If $u(p)$ is the head of $p'$, $p'$ still occupies $u(p)$ in $c_{i+1}$.
 	Otherwise, $p'$ reads a follow indicator from $p$ over a tail edge in $c_i$
 	and therefore the contraction must be part of a handover.
 	As $p$ is not involved in the action, the handover must be between $p'$ and a third active particle $p''$.
 	It is easy to see that after such a handover $u(p)$ is occupied by either $p'$ or $p''$.
	Finally, if $a$ affects neither $p$ nor $p'$, $u(p)$ will still be occupied by $p'$ in $c_{i+1}$.
\end{proof}

Based on Lemma~\ref{lem:successor}, we define a successor relation on the active particles in a configuration $c$.
Let $p$ be a follower.
We say $p'$ is the \emph{successor} of $p$ if $p'$ occupies $u(p)$.
Analogously, we say $p$ is a \emph{predecessor} of $p'$.
Furthermore, we define a directed graph $A(c)$ for a configuration $c$ as follows.
$A(c)$ contains the same nodes as $G(c)$.
For every expanded particle $p$ in $c$, $A(c)$ contains a directed edge from the tail to the head of $p$,
and for every follower $p'$ in $c$, $A(c)$ contains a directed edge from the head of $p'$ to $u(p')$.

\begin{lemma}
	\label{lem:forest}
	The graph $A(c)$ is a forest, and if there is at least one active particle,
    every connected component of inactive particles contains a particle that is connected to an active particle.
\end{lemma}
\begin{proof}
    In an initial configuration $c_0$, all particles are inactive and therefore the lemma holds trivially.
    Now assume that the lemma holds for a configuration $c_i$.
    We will show that it also holds for the next configuration $c_{i+1}$ that results from executing an action $a$.
    If $a$ affects an inactive particle $p$, this particle either becomes a follower or a leader.
	In the former case $p$ joins an existing tree, and in the latter case $p$ forms a new tree in $A(c_{i+1})$.
    In either case, $A(c_{i+1})$ is a forest and the connected component of inactive particles
    that $p$ belongs to in $c_i$ is either non-existent or connected to $p$ in $c_{i+1}$.
	If $a$ affects only a single particle $p$ that is in phase follow, this particle can contract or become a leader.
	In the former case, $p$ has no predecessor $p'$ such that $u(p')$ is the tail of $p$
	and also $p$ has no idle neighbors.
	Therefore, the contraction of $p$ does not disconnect any follower or inactive particle
	and, accordingly, does not violate the lemma.
	In the latter case, $p$ becomes a root of a tree which also does not violates the lemma.
    If $a$ involves only a single particle $p$ that is in phase lead, $p$ can expand or contract.
    An expansion trivially cannot violate the lemma and the argument for the contraction is the same as for the contraction of a follower above.
    Finally, if $a$ involves two active particles in $c_i$, these particles perform a handover.
    While such a handover can change the successor relation among the nodes, it cannot violate the lemma.
\end{proof}

The following lemma shows that the spanning forest algorithm achieves progress
in that as long as the leaders keep moving, the remaining particles will eventually follow them.

\begin{lemma}
	\label{lem:contract}
	An expanded particle eventually contracts.
\end{lemma}
\begin{proof}
	Consider an expanded particle $p$ in a configuration $c$.
	Note that $p$ must be active.
	If there is an enabled action that includes the contraction of $p$,
	that action will remain enabled until $p$ contracts and therefore $p$ will contract eventually
	according to the fairness assumption we made in Section~\ref{sec:problemDefinition}.
	So assume that there is no enabled action that includes the contraction of $p$.
	According to the behavior of inactive particles,
	at some point in time all particles in the system will be active.
	If the contraction of $p$ becomes part of an enabled action before this happens, $p$ will eventually contract.
	So assume that all particles are active but still $p$ cannot contract.
	If $p$ has no predecessors, the isolated contraction of $p$ is an enabled action which contradicts our assumption.
	Therefore, $p$ must have predecessors.
	Furthermore, $p$ must read at least one follow indicator over a tail edge and all predecessors
	from which it reads a follow indicator must be expanded
	as otherwise $p$ could again contract as part of a handover.
	Let $p'$ be one of the predecessors of $p$.
	If $p'$ would contract, a handover between $p'$ and $p$ would become an enabled action.
	We can apply the complete argument presented in this proof so far to $p'$
	and so on backwards along a branch in a tree in $A(c)$ until we reach a particle that can contract.
	We will reach such a particle by Lemma~\ref{lem:forest}.
	Therefore, we found a sequence of expanded particles that starts with $p'$
	and ends with a particle that eventually contracts.
	The contraction of that last particle will allow the particle before it in the sequence to contract and so on.
	Finally, the contraction of $p$ will become part of an enabled action and therefore $p$ will eventually contract.
\end{proof}

In the above lemmas, the direction of expansion of leaders is not used.
Furthermore, the fact that only particles on the surface become leaders is not used.
Therefore, the algorithm works independently of the selection of leaders and their expansion direction.
This makes the spanning forest algorithm a reusable algorithmic primitive.

\subsubsection{Complaining Algorithm}
\label{sec:complainingAlgorithm}
\vspace{-2mm}
The algorithm so far achieves that the particles spread out towards one direction on the surface,
which will be shown formally in Section~\ref{sec:analysis}.
However, the particles keep moving indefinitely even when all particles lie on the surface.
Since we require termination from an algorithm to solve the Infinite Object Coating problem,
we need another algorithmic primitive that ensures that once all particles are on the surface,
they eventually stop moving.
To achieve this, we use the idea of \emph{complaining}, see Algorithm~\ref{alg:complainingAlgorithm}.
The algorithm extends Algorithm~\ref{alg:spanningForestAlgorithm} by changing the set of turns for leaders.
The conditions in Algorithm~\ref{alg:complainingAlgorithm} ought to be checked
before the conditions given in Algorithm~\ref{alg:spanningForestAlgorithm}.

\begin{algorithm}
	Consider a leader particle $p$ and let $s$ be the direction returned by Algorithm~\ref{alg:FaSAlgorithm}, i.e., the direction that leaders use to travel along the surface.
	Leaders can include a \emph{complaint indicator} in a flag.
	If $p$ is contracted and cannot expand or perform a handover and sees a follow indicator or complaint indicator, it sends a complaint indicator in direction $s$ and performs the idle movement.
	If $p$ is contracted 
	and does not see a complaint indicator, it does not perform any action.
	Otherwise $p$ behaves according to Algorithm~\ref{alg:spanningForestAlgorithm}.
	\caption{Complaining Algorithm}
	\label{alg:complainingAlgorithm}
\end{algorithm}

Note that a complaint indicator will be \emph{consumed} by a leader $p$ if it expands, contracts, or performs a handover. That is, as long as all particles which forwarded the indicator have not moved up to $p$, $p$ will not see a complaint indicator. Furthermore, consider a follower $q$ that reached the surface, but is not a leader yet. If $q$ reads a complaint indicator, it will not forward the indicator directly, but as soon as it turns into a leader.  Moreover, if all particles are leaders, then no leader sees a follow indicator.
We extend the notion of $u(p)$ from Section~\ref{sec:spanningForestAlgorithm} to leaders. The node $u(p)$ for a leader $p$ is the node in the direction returned by  Algorithm~\ref{alg:FaSAlgorithm}. Hence, the notion of successors (i.e., $p'$ is a successor of $p$ in some configuration $c$ if $p'$ occupies $u(p)$) is now also applicable for leaders. If $u(p)$ is unoccupied, $p$ has no successor. 
The \emph{descendants} of a particle $p$ are all nodes reachable by the successor relation (i.e., the successor of, the successor of the successor, and so on). For each particle $p$ we denote the descendant that has no successor with $a(p)$.
For the next lemma consider a system that behaves according to Algorithm~\ref{alg:spanningForestAlgorithm} and Algorithm~\ref{alg:complainingAlgorithm}.

\begin{lemma} 
\label{lem:complain}
As long as a follower particle $p$ exists, a descendant will eventually expand, and if all particles are leaders the transition function of every leader eventually maps to the empty set.
\end{lemma}

\begin{proof}
For the first statement assume that even though the particle $p$ exists, no descendant expands. 
Following our assumption none of the descendants can expand, therefore they all have to be contracted, because an expanded descendant would allow for a handover which involves an expansion.
Therefore, a complaint indicator is created by a leader particle that is a descendant of $p$ and sees a follow indicator. This indicator is forwarded among the descendants along the surface until $a(p)$ sees it. 
Particle $a(p)$ can always expand, which contradicts the assumption. 

To prove the second statement, we look at the case in which all particles are leaders. We already mentioned that no more follow indicators exist. Therefore, it is easy to see that all complaint indicators eventually vanish. Accordingly, the transition function of all particles without a neighbor in direction $s$ maps to the empty set. As a result, the transition function of leaders that are neighbors to leaders without a neighbor in direction $s$ will eventually map to the empty set. This process continues until the transition function of every particle maps to the empty set.
\end{proof}

\subsection{Analysis}
\label{sec:analysis}
\vspace{-2mm}
Now, we can show that our algorithm as developed in the previous three sections
solves the Infinite Object Coating problem.

\begin{theorem}
	\label{thm:solve}
	Our algorithm solves the Infinite Object Coating problem.
\end{theorem}
\begin{proof}
	First, we have to show that the algorithm maintains connectivity.
	So consider a system of particles in a configuration during the execution of our algorithm.
	The object is by definition connected.
	A leader always lies on the surface of the object according to Algorithm~\ref{alg:spanningForestAlgorithm}.
	A follower is always part of a tree in the spanning forest as shown in Lemma~\ref{lem:forest}.
	As every tree forms a connected component and is rooted in a leader,
	the set of object particles and active particles forms a connected component.
	Finally, an inactive particle is always part of a connected component of inactive particles
	that includes a particle that is connected to an active particle, again by Lemma~\ref{lem:forest}.
	Therefore, all particles in the system form a single connected component.

	Next, we have to show that the algorithm terminates and that when it does, the system is in a goal configuration.
	A common property of all goal configurations is that all particles lie on the surface.
	In our algorithm, every particle $p$ eventually activates.
	If $p$ initially lies on the surface, it becomes a leader and remains on the surface.
	If $p$ initially does not lie on the surface, it becomes a follower.
	Let $c$ be the first configuration in which $p$ is a follower.
	Consider the directed path in $A(c)$ from the head of $p$ to the first node on the surface.
	There always is such a path since every follower belongs to a tree in $A(c)$ by Lemma~\ref{lem:forest},
	every such tree is rooted in a leader, and a leader only occupies nodes on the surface.
	Let $P = (u_0, u_1, \ldots, u_k)$ be that path where $u_0$ is the head of $p$ and $u_k$ lies on the surface.
	According to Algorithm~\ref{alg:spanningForestAlgorithm}, $p$ attempts to follow $P$
	by sequentially	expanding into the nodes $u_1, \ldots, u_k$.
	By Lemma~\ref{lem:complain}, the algorithm does not terminate before $p$ reaches the surface,
	and according to Lemma~\ref{lem:contract}, $p$ can actually execute all movements required to follow $P$.
	Therefore, $p$ eventually lies on the surface, becomes a leader, and remains on the surface.
	According to Lemma~\ref{lem:complain} this means that for all particles the transition function
	eventually maps to the empty set  which implies termination.
\end{proof}

Finally, we would like to measure how well our algorithm performs in terms of energy consumption.
For this, we consider the number of movements executed in a system until termination
and call this measure \emph{work}.
When we refer to movement in the context of work, we only mean expansions and contraction but not idle movements.
We count a handover as two movements.
We ignore any computation a particle performs since in a physical realization the energy consumption of computation
is most likely negligible compared to the energy consumption of movement.

\begin{lemma}
	\label{lem:worstCase}
	The worst-case work required by any algorithm to solve the Infinite Object Coating problem is $\Omega(n^2)$.
\end{lemma}
\begin{proof}
	Consider the configuration depicted in Figure~\ref{fig:worstCase}.
	\begin{figure}
	    \centering
	    \includegraphics[scale=0.85]{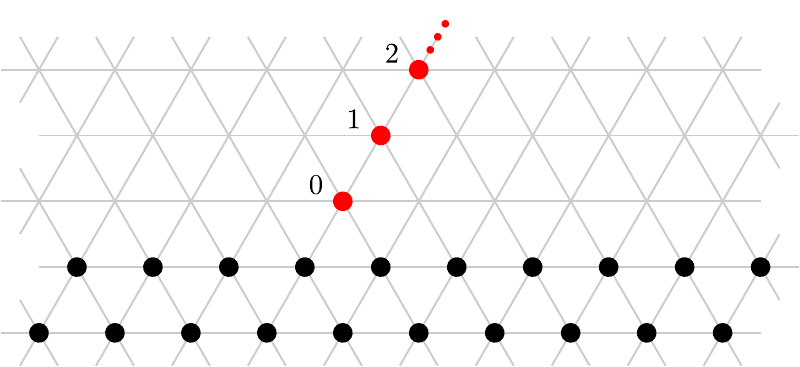}
	    \caption{\small
	    	Worst-case configuration concerning work.
	    	The object particles are shown in black and the non-object particles are shown in red.
	    	The infinite object is a half-plane and the $n$ non-object particles lie on a straight line.
	    }
	    \label{fig:worstCase}
	\end{figure}
	The particle labeled $i$ requires at least $2i$ movements before it lies contracted on the surface.
	Therefore, any algorithm requires at least $\sum_{i = 0}^{n-1} 2i = \Omega(n^2)$ work.
\end{proof}

\begin{theorem}
	Our algorithm requires worst-case optimal work $\Theta(n^2)$.
\label{thm:work}
\end{theorem}
\begin{proof}
	To prove the upper bound, we simply show that every particle executes $O(n)$ movements.
	The theorem then follows by Lemma~\ref{lem:worstCase}.
	Consider a particle $p$.
	While $p$ is inactive, it does not move.
	While $p$ is a follower, it moves along a path to the surface as described in the proof of Theorem~\ref{thm:solve}.
	The length of this path is bound by $2n$ and, therefore, the number of movements $p$ executes while being a follower is $O(n)$.
	While $p$ is a leader it only performs expansions if it reads a complaint indicator. Since a complaint indicator is consumed by an expansion (see Section~\ref{sec:complainingAlgorithm}),
	a leader can see at most $n-1$ indicators. Every expansion is followed by a contraction, 
	therefore the number of movements $p$ executes while being a leader is as well $O(n)$, which concludes the theorem.
\end{proof}

\section{Conclusion}
\label{sec:researchChallenges}
\vspace{-2mm}
In this work we have formally defined the Amoebot model and presented a work-optimal algorithm for the Infinite Object Coating problem under this model.
We want to use the Amoebot model to investigate various other problems in which
 the system of particles forms a single connected component at all times. Other \emph{coating problems} might be considered, in particular when the object surface is finite and the surface of an object has to be coated as uniformly as possible by the particles of a system (possibly with multiple layers of ``coating'').
A second example is the class of \emph{shape formation problems}
in which a system has to arrange to form a specific shape, with or without a seed particle.
Finally, in \emph{bridging problems} particles have to bridge gaps in given structures.
We see the coating as an algorithmic primitive  for solving other problems.
For example, the formation of a shape can be achieved by creating an initially small instance of that shape
which is then iteratively coated to form increasingly large instances
until the number of particles in the system is exhausted. Furthermore, we envision that our spanning forest algorithm (Section~\ref{sec:spanningForestAlgorithm}) may in turn be a building block for other variations of the coating problems.


\bibliographystyle{plain}
\bibliography{literature}

\begin{thebibliography}{10}

\bibitem{Adl94}
L.~M. Adleman.
\newblock Molecular computation of solutions to combinatorial problems.
\newblock {\em Science}, 266(11):1021--1024, 1994.

\bibitem{AG3}
Chrysovalandis Agathangelou, Chryssis Georgiou, and Marios Mavronicolas.
\newblock A distributed algorithm for gathering many fat mobile robots in the
  plane.
\newblock In {\em Proceedings of the 2013 ACM symposium on Principles of
  distributed computing}, pages 250--259. ACM, 2013.

\bibitem{AE07}
R.~Ananthakrishnan and A.~Ehrlicher.
\newblock The forces behind cell movement.
\newblock {\em International Journal of Biological Sciences}, 3(5):303--317,
  2007.

\bibitem{AAD+06}
D.~Angluin, J.~Aspnes, Z.~Diamadi, M.~J. Fischer, and R.~Peralta.
\newblock Computation in networks of passively mobile finite-state sensors.
\newblock {\em Distributed Computing}, 18(4):235--253, 2006.

\bibitem{AR10}
D.~Arbuckle and A.~Requicha.
\newblock Self-assembly and self-repair of arbitrary shapes by a swarm of
  reactive robots: algorithms and simulations.
\newblock {\em Autonomous Robots}, 28(2):197--211, 2010.

\bibitem{BA99}
Albert-Laszlo Barabasi and Reka Albert.
\newblock Emergence of scaling in random networks.
\newblock {\em Science}, 286(5439):509--512, 1999.

\bibitem{BFMS11}
L.~Barriere, P.~Flocchini, E.~Mesa-Barrameda, and N.~Santoro.
\newblock Uniform scattering of autonomous mobile robots in a grid.
\newblock {\em Int. Journal of Foundations of Computer Science},
  22(3):679--697, 2011.

\bibitem{BPEA+01}
Y.~Benenson, T.~Paz-Elizur, R.~Adar, E.~Keinan, Z.~Livneh, and E.~Shapiro.
\newblock Programmable and autonomous computing machine made of biomolecules.
\newblock {\em Nature}, 414(6862):430--434, 2001.

\bibitem{DBLP:1211-1909}
A.~Bhattacharyya, M.~Braverman, B.~Chazelle, and H.L. Nguyen.
\newblock On the convergence of the hegselmann-krause system.
\newblock {\em CoRR}, abs/1211.1909, 2012.

\bibitem{BDLS96}
D.~Boneh, C.~Dunworth, R.~J. Lipton, and J.~Sgall.
\newblock On the computational power of {DNA}.
\newblock {\em Discrete Applied Mathematics}, 71:79--94, 1996.

\bibitem{BMV12}
V.~Bonifaci, K.~Mehlhorn, and G.~Varma.
\newblock Physarum can compute shortest paths.
\newblock In {\em Proceedings of SODA '12}, pages 233--240, 2012.

\bibitem{BKRT04}
Z.~J. Butler, K.~Kotay, D.~Rus, and K.~Tomita.
\newblock Generic decentralized control for lattice-based self-reconfigurable
  robots.
\newblock {\em International Journal of Robotics Research}, 23(9):919--937,
  2004.

\bibitem{Cha09}
B.~Chazelle.
\newblock Natural algorithms.
\newblock In {\em Proc. of ACM-SIAM SODA}, pages 422--431, 2009.

\bibitem{CDBG11}
K.~C. Cheung, E.~D. Demaine, J.~R. Bachrach, and S.~Griffith.
\newblock Programmable assembly with universally foldable strings (moteins).
\newblock {\em IEEE Transactions on Robotics}, 27(4):718--729, 2011.

\bibitem{Chi94}
G.~Chirikjian.
\newblock Kinematics of a metamorphic robotic system.
\newblock In {\em Proceedings of ICRA '94}, volume~1, pages 449--455, 1994.

\bibitem{ci12}
Mark Cieliebak, Paola Flocchini, Giuseppe Prencipe, and Nicola Santoro.
\newblock Distributed computing by mobile robots: Gathering.
\newblock {\em SIAM Journal on Computing}, 41(4):829--879, 2012.

\bibitem{CP08}
R.~Cohen and D.~Peleg.
\newblock Local spreading algorithms for autonomous robot systems.
\newblock {\em Theoretical Computer Science}, 399(1-2):71--82, 2008.

\bibitem{DFSY10}
S.~Das, P.~Flocchini, N.~Santoro, and M.~Yamashita.
\newblock On the computational power of oblivious robots: forming a series of
  geometric patterns.
\newblock In {\em Proceedings of 29th ACM Symposium on Principles of
  Distributed Computing (PODC)}, 2010.

\bibitem{DS08}
X.~Defago and S.~Souissi.
\newblock Non-uniform circle formation algorithm for oblivious mobile robots
  with convergence toward uniformity.
\newblock {\em Theoretical Computer Science}, 396(1-3):97--112, 2008.

\bibitem{DPSS11}
E.~D. Demaine, M.~J. Patitz, R.~T. Schweller, and S.~M. Summers.
\newblock Self-assembly of arbitrary shapes using rnase enzymes: Meeting the
  kolmogorov bound with small scale factor (extended abstract).
\newblock In {\em Proceedings of STACS '11}, pages 201--212, 2011.

\bibitem{DBLP:conf/spaa/DerakhshandehDGRSS14}
Zahra Derakhshandeh, Shlomi Dolev, Robert Gmyr, Andr{\'{e}}a~W. Richa,
  Christian Scheideler, and Thim Strothmann.
\newblock Brief announcement: amoebot - a new model for programmable matter.
\newblock In {\em 26th {ACM} Symposium on Parallelism in Algorithms and
  Architectures, {SPAA} '14, Prague, Czech Republic - June 23 - 25, 2014},
  pages 220--222, 2014.

\bibitem{DolevGRS13}
Shlomi Dolev, Robert Gmyr, Andréa~W. Richa, and Christian Scheideler.
\newblock Ameba-inspired self-organizing particle systems.
\newblock {\em CoRR}, abs/1307.4259, 2013.

\bibitem{DHKS12}
Maximilian Drees, Martina H{\"u}llmann, Andreas Koutsopoulos, and Christian
  Scheideler.
\newblock Self-organizing particle systems.
\newblock In {\em IPDPS}, pages 1272--1283, 2012.

\bibitem{em13}
Yuval Emek and Roger Wattenhofer.
\newblock Stone age distributed computing.
\newblock In {\em Proceedings of the 2013 ACM symposium on Principles of
  distributed computing}, pages 137--146. ACM, 2013.

\bibitem{FGK10}
S.~Fekete, C.~Gray, and A.Kroeller.
\newblock Evacuation of rectilinear polygons.
\newblock In {\em Proceedings of the 4th International Conference Combinatorial
  Optimization and Applications (COCOA)}, pages 21--30, 2010.

\bibitem{FS10}
S.P. Fekete and C.~Schmidt.
\newblock Polygon exploration with time-discrete vision.
\newblock {\em Computational Geometry}, 43(2):148--168, 2010.

\bibitem{fl13}
Paola Flocchini, David Ilcinkas, Andrzej Pelc, and Nicola Santoro.
\newblock Computing without communicating: Ring exploration by asynchronous
  oblivious robots.
\newblock {\em Algorithmica}, 65(3):562--583, 2013.

\bibitem{fl08}
Paola Flocchini, Giuseppe Prencipe, Nicola Santoro, and Peter Widmayer.
\newblock Arbitrary pattern formation by asynchronous, anonymous, oblivious
  robots.
\newblock {\em Theoretical Computer Science}, 407(1):412--447, 2008.

\bibitem{FNKB88}
T.~Fukuda, S.~Nakagawa, Y.~Kawauchi, and M.~Buss.
\newblock Self organizing robots based on cell structures - cebot.
\newblock In {\em Proceedings of IROS '88}, pages 145--150, 1988.

\bibitem{HABFM02}
T.-R. Hsiang, E.~Arkin, M.~Bender, S.~Fekete, and J.~Mitchell.
\newblock Algorithms for rapidly dispersing robot swarms in unknown
  environments.
\newblock In {\em Proceedings of the 5th Workshop on Algorithmic Foundations of
  Robotics (WAFR)}, pages 77--94, 2002.

\bibitem{Kat05}
B.~Katreniak.
\newblock Biangular circle formation by asynchronous mobile robots.
\newblock In {\em Proceedings of the 12th International Colloquium on
  Structural Information and Communication Complexity (SIROCCO)}, pages 85--99,
  2005.

\bibitem{Ker12}
S.~Kernbach, editor.
\newblock {\em Handbook of Collective Robotics -- Fundamentals and Challanges}.
\newblock Pan Stanford Publishing, 2012.

\bibitem{Kle00}
J.~Kleinberg.
\newblock The small-world phenomenon: an algorithmic perspective.
\newblock In {\em Proceedings of STOC '00}, pages 163--170, 2000.

\bibitem{KM11}
P.~Kling and F.~{Meyer auf der Heide}.
\newblock Convergence of local communication chain strategies via linear
  transformations.
\newblock In {\em Proceedings of the 23rd ACM Symposium on Parallelism in
  Algorithms and Architectures}, pages 159--166, 2011.

\bibitem{KSP12}
T.~Krasinski, S.~Sakowski, and T.~Poplawski.
\newblock Autonomous push-down automaton built on dna.
\newblock {\em Informatica}, 36:263--276, 2012.

\bibitem{LTT+10}
K.~Li, K.~Thomas, C.~Torres, L.~Rossi, and C.-C. Shen.
\newblock Slime mold inspired path formation protocol for wireless sensor
  networks.
\newblock In {\em Proceedings of ANTS '10}, pages 299--311, 2010.

\bibitem{McL08}
J.~McLurkin.
\newblock {\em Analysis and Implementation of Distributed Algorithms for
  Multi-Robot Systems}.
\newblock PhD thesis, Massachusetts Institute of Technology, 2008.

\bibitem{NKC03}
R.~Nagpal, A.~Kondacs, and C.~Chang.
\newblock Programming methodology for biologically-inspired self-assembling
  systems.
\newblock Technical report, AAAI Spring Symposium on Computational Synthesis,
  2003.

\bibitem{PZ06}
C.~Parker and H.~Zhang.
\newblock Collective robotic site preparation.
\newblock {\em Adaptive Behavior}, 14(1):5--19, 2006.

\bibitem{RS10}
M.~Rubenstein and W.~Shen.
\newblock Automatic scalable size selection for the shape of a distributed
  robotic collective.
\newblock In {\em Proc.~of the IEEE/RSJ Intl.~Conf.~on Intelligent Robots and
  Systems (IROS)}, 2010.

\bibitem{st11}
Aaron Sterling.
\newblock Distributed agreement in tile self-assembly.
\newblock {\em Natural Computing}, 10(1):337--355, 2011.

\bibitem{st08}
Aaron~D Sterling.
\newblock A limit to the power of multiple nucleation in self-assembly.
\newblock In {\em Distributed Computing}, pages 451--465. Springer, 2008.

\bibitem{SY99}
I.~Suzuki and M.~Yamashita.
\newblock Distributed anonymous mobile robots: Formation of geometric patterns.
\newblock {\em SIAM Journal on Computing}, 28(4):1347--1363, 1999.

\bibitem{to91}
Tommaso Toffoli and Norman Margolus.
\newblock Programmable matter: concepts and realization.
\newblock {\em Physica D: Nonlinear Phenomena}, 47(1):263--272, 1991.

\bibitem{WWA04}
J.~E. Walter, J.~L. Welch, and N.~M. Amato.
\newblock Distributed reconfiguration of metamorphic robot chains.
\newblock {\em Distributed Computing}, 17(2):171--189, 2004.

\bibitem{WTTN11}
S.~Watanabe, A.~Tero, A.~Takamatsu, and T.~Nakagaki.
\newblock Traffic optimization in railroad networks using an algorithm
  mimicking an amoeba-like organism, physarum plasmodium.
\newblock {\em Biosystems}, 105(3):225--232, 2011.

\bibitem{WS98}
D.~J. Watts and S.~H. Strogatz.
\newblock Collective dynamics of 'small-world' networks.
\newblock {\em Nature}, 393(6684):440--442, 1998.

\bibitem{WLWS98}
E.~Winfree, F.~Liu, L.~A. Wenzler, and N.~C. Seeman.
\newblock Design and self-assembly of two-dimensional dna crystals.
\newblock {\em Nature}, 394(6693):539--544, 1998.

\bibitem{winfree13}
Damien Woods, Ho-Lin Chen, Scott Goodfriend, Nadine Dabby, Erik Winfree, and
  Peng Yin.
\newblock Active self-assembly of algorithmic shapes and patterns in
  polylogarithmic time.
\newblock In {\em ITCS}, pages 353--354, 2013.

\bibitem{YSS+07}
M.~Yim, W.-M. Shen, B.~Salemi, D.~Rus, M.~Moll, H.~Lipson, E.~Klavins, and
  G.~S. Chirikjian.
\newblock Modular self-reconfigurable robot systems.
\newblock {\em IEEE Robotics Automation Magazine}, 14(1):43--52, 2007.

\end{thebibliography}

\end{document}